\documentclass[11pt]{amsart}
\usepackage{graphicx}

\usepackage{amsmath,amsfonts,amssymb}

\usepackage{xcolor}

\newtheorem{theorem}{Theorem}
\newtheorem{lemma}[theorem]{Lemma}

\newtheorem{definition}[theorem]{Definition}
\newtheorem{corollary}[theorem]{Corollary}
\newtheorem{remark}[theorem]{Remark}

\newtheorem{example}[theorem]{Example}

\newtheorem{procedure}{Procedure}
\newtheorem{algorithm}{Algorithm}

\newcommand{\N}{{\mathbb N}}
\newcommand{\Z}{{\mathbb Z}}
\newcommand{\F}{\mathbb F}
\newcommand{\fF}{\mathfrak F}
\newcommand{\Le}{\mathbb L}
\newcommand{\A}{{\mathcal A}}
\newcommand{\B}{{\mathcal B}}
\newcommand{\D}{{\mathcal D}}
\newcommand{\R}{{\mathcal R}}

\newcommand{\bL}{\mathbf {\Lambda}}

\newcommand{\Row}{\mathrm{Row}}

\newcommand{\tq}{\, \mid \,}
 
\newcommand{\supp}{{\rm supp}}   

\begin{document}
\title[Decoding Hyperbolic-like Codes by BMS]{Decoding up to $4$ errors in Hyperbolic-like Abelian Codes by the Sakata Algorithm}

\author[J.J. Bernal]{Jos\'{e} Joaqu\'{\i}n Bernal}

\email{josejoaquin.bernal@alu.um.es jsimon@um.es}

\author[J.J. Sim\'{o}n]{Juan Jacobo Sim\'{o}n}
\address{Departamento de Matem\'{a}ticas, Universidad de Murcia, Espa\~{n}a}
\thanks{This work was partially supported by MINECO, project MTM2016-77445-P, and Fundaci\'{o}n S\'{e}neca of Murcia, project 19880/GERM/15.}

\maketitle              

\begin{abstract}
We deal with two problems related with the use of the Sakata's algorithm in a specific class of bivariate codes (see \cite{Blah,Sakata 2,Sakata}). The first one is to improve the general framework of locator decoding in order to apply it on such abelian codes. The second one is to find  a set of indexes oF the syndrome table such that no other syndrome contributes to implement the BMSa and, moreover, any of them may be ignored \textit{a priori}. In addition, the implementation on those indexes is sufficient to get the Groebner basis; that is, it is also a termination criterion.

\end{abstract}
\section{Introduction}
The Sakata algorithm (or Berlekamp-Massey-Sakata algorithm, BMSa, for short) is one of the best known procedures to find Groebner basis (see \cite{Cox,Sakata 6}) for the so called ideal of linear recurrence relations on a doubly periodic array  \cite{rubio,Sakata 2}. It is a common method for decoding algebraic geome\-tric  codes, specially those constructed from one-point algebraic curves \cite{Blah,Cox et al Using,Sakata 3} and, within them, the family of Hyperbolic Cascade Reed-Solomon Codes (see \cite{saints heegard}). Less studied or understood is the original application of the BMSa: decoding Abelian Codes through locator decoding \cite{Sakata}.

The general idea of decoding with the BMSa is as follows: a codeword of an abelian code of lenght $r=r_1\cdot r_2$ (see notation below), say $c=c(X_1,X_2)$ was sent and we receive $f=c+e$; where $e=e(X_1,X_2)$ is called the error polynomial, as usual. To find $e$, we consider what we will call the syndorme values $e(\alpha_1^i,\alpha_2^j)=s_{i,j}$, where $\alpha_i$ is a $r_i$-th primitive root of unity, for $i=1,2$. It is clear that we only have to know the syndrome values for $0\leq i\leq r_1-1$ and $0\leq j\leq r_2-1$; that is,  $S=(s_{ij})_{\N\times\N}$ is a doubly $r_1\times r_2$-periodic array (see Definition~\ref{Doubly period array}(1)). In fact, as $e(X_1,X_2)$ is unknown, \textit{a priori} we will not be able to know the entire table $S$, meanwhile, as we will see, some unknown syndromes values may be discovered by using the specific properties of the given abelian code. Locator decoding shows us a close relation between the error positions and the ideal of linear recurring relations of $S$, $\bL(S)$, in such a way that if we find a Groebner basis for it we may know such positions. 

The BMSa is an iterative procedure with respect to a given well-ordering on $\N\times\N$ to find the mentioned Groebner basis. On each step, there is a given set of polynomials, called minimal set of polynomials, say $F_l$, which is updated (possibly no strict) to a new minimal set of polynomials, $F_{l+1}$; where $l+1$ means the successor with respect to the well ordering.  The process is implemented until we may apply some termination criterion in certain step (see \cite{Cox et al Using,Hackl,rubio,Sakata 2}). This criterion is (always) based on the shape of all possible footprints (see \cite[p. 1615]{Blah}) of $\bL(S)$ that we might obtain in such step.

During the implementation of the BMSa, it may happen that some  iteration does not update the given minimal set of polynomials; that is, the equality between $F_l$ and $F_{l+1}$ holds. In this context, the main goal of our paper is to prove that, up to 4 errors, there is a set of indexes, say $\B$, such that if $l\not\in \B$ then $F_l=F_{l+1}$ (that is, no other syndrome contributes to construct the Groebner basis) and, moreover, none of them may be ignored \textit{a priori}; that is, $\B$ is minimal with that property. In addition, our condition is a termination criterion because it is known that the Groebner basis may be always be obtained if one cover enough steps. 

On the other hand, a characteristic of the BMSa is that one should begin at the place $(0,0)$. This is a natural condition when codes are defined over splitting fields, as for example, Hyperbolic Cascade Reed-Solomon codes; however, in general, this may not happen, specially in the case of binary abelian codes, as in Example~\ref{ejemplo de (0,t+j)},below. This is another goal of our paper: we give an improvement of the framework for applying the locator decoding algorithm in such a way that we may consider a translated table that allows us to start in a initial point different from (0,0). This paper is a portion of a study in progress of the BMSa in a more general framework.

\section{Bivariate codes}

Let $\F$  be a finite field with $q$ elements, with $q$ a power of a prime number, let $r_i$ be positive integers, for $i\in \{ 1,2\}$, and $r=r_1\cdot r_2$.  We denote by $\Z_{r_i}$ the ring of integers modulo $r_i$. We always write its elements as canonical representatives. When necessary, we write $\overline{a}\in\Z_k$ for any $a\in\Z$ and $k\in \N$.
 
A \textbf{bivariate code}, or 2-dimensional abelian code, of length $r$ (see \cite{Imai}) is an ideal in the algebra   $\F(r_1, r_2)=\F[X_1,X_2]/\langle X_1^{r_1}-1, X_2^{r_2}-1\rangle $. Throughout this work,  we assume that this algebra is semisimple; that is, $\gcd (r_i,q)=1$, for  $i\in \{ 1,2\}$.  The codewords are identified with polynomials. The weight of a codeword $c$ is denoted by $\omega(c)$. We denote by $I$ the set $\Z_{r_1}\times \Z_{r_2}$ and we  write the elements $f \in  \F(r_1,r_2)$ as $f=\sum a_m \mathbf{X}^m$, where $m=(m_1,m_2)\in I$ and $\mathbf{X}^m=X_1^{m_1}\cdot X_2^{m_2}$. Given a polynomial $f \in \F[X_1,X_2]$, we denote by $\overline{f}$ its image under the canonical projection onto $\F(r_1,r_2)$, when necessary.

For each $i\in \{ 1,2\}$, we denote by $R_{r_i}$ (resp. $\R_{r_i}$) the set of  $r_i$-th roots of unity (resp.  $r_i$-th primitive roots of unity) and define $R=R_{r_1}\times R_{r_2}$ ($\R=\R_{r_1}\times \R_{r_2}$). Throughout this paper, we fix $\Le|\F$ as a extension field containing $R_{r_i}$.

For $f=f(X_1,X_2) \in \F[X_1,X_2]$ and $\boldsymbol{\alpha}=(\alpha_1,\alpha_2)\in R$, we write $f(\boldsymbol{\alpha})=f(\alpha_1,\alpha_2)$. For $m=(m_1,m_2)\in I$, we write $\boldsymbol{\alpha}^{m} = (\alpha_1^{m_1},\alpha_2^{m_2})$.

It is a known fact that, in the semi simple case, every abelian code $C$ in $\F(r_1,r_2)$ is totally determined by its \textbf{root set} or \textbf{set of zeros}, namely
$$Z(C)=\left\{\boldsymbol{\alpha}\in  R \tq f(\boldsymbol{\alpha})=0,\;\; \mbox{ for all }\; f\in C \right\}.$$  
For a fixed $\boldsymbol{\alpha}\in \R$, the code $C$ is  determined by its \textbf{defining set}, with respect to $\boldsymbol{\alpha}$, which is defined as 
$$\D_{\boldsymbol{\alpha}}(C) = \left\{ m\in I \tq \boldsymbol{\alpha}^{m}\in Z(C)\right\}.$$ 
It is easy to see that the notions of set of zeros and defining set may be considered for any set of either polynomials or ideals in $\F(r_1,r_2)$ (or $\Le(r_1,r_2)$); moreover, it is known that for any $G\subset \F(r_1,r_2)$ (or $\Le(r_1,r_2)$) and $\boldsymbol \alpha \in \R$ we have $D_{\boldsymbol{\alpha}}(G)=D_{\boldsymbol{\alpha}}(\langle G\rangle)$. In \cite{Blah,Cox}, the defining set is  considered for ideals $P$ in $\Le[\mathbf{X}]$. From the definition, we have $D_{\boldsymbol{\alpha}}(P)=D_{\boldsymbol{\alpha}}(\overline{P})$, where $\overline{P}$ is the canonical projection of $P$ onto $\Le(r_1,r_2)$.
 
We also recall the extension of the concept of $q$-cyclotomic coset of an integer to two components. 

Given an element $(a_1,a_2)\in I$, we define its \textit{$q$-orbit} modulo  $\left(r_1,r_2\right)$ as
	\begin{equation}\label{qorbita}
	  Q(a_1,a_2)=\left\{\left(a_1\cdot q^i ,a_2\cdot q^i  \right)\tq i\in \N\right\} \subseteq I= \Z_{r_1}\times\Z_{r_2}.
	\end{equation}
	
It is easy to see that for every abelian code $C\subseteq\F(r_1,r_2)$, $\D_{\boldsymbol{\alpha}}\left(C\right)$ is closed under multiplication by $q$ in $I$, and then $\D_{\boldsymbol{\alpha}}(C)$ is necessarily a disjoint union of $q$-orbits modulo $(r_1,r_2)$. Conversely, every union of $q$-orbits modulo $(r_1,r_2)$ defines an abelian code in $\F(r_1,r_2)$. For the sake of simplicity we only write $q$-orbit, and the tuple of integers will be clear from the context.

\section{Apparent distance and multilevel bound}

In \cite[p. 1614]{Blah}, Blahut introduce the Hyperbolic 
codes of designed distance $\delta$; a purely algebraic version of one-point AG-codes mentioned in Introduction. They are abelian codes whose defining set with respect to $\boldsymbol\alpha \in \R$ is 
 $$\A=\D_{\boldsymbol{\alpha}}(C)=\left\{(i,j)\in I\tq (i+1)(j+1)\leq \delta\right\}.$$
 In practice, the computation of the syndrome values is done over the defining set. On the other hand, for this family of codes it is known that $\delta$ (called the multilevel bound) is a lower bound for the minimum distance of $C$, denoted by $d(C)$ .
 
 Here, we point out that we use sets of the form $\A$ for both AG-codes and Hyperbolic codes in two directions. One of them is to establish their multilevel bounds and the second one is to set it as a base of syndromes from which one may try to infer those extra syndrome values needed to implement steps in the BMSa; however, as we mentioned in the Introduction, the BMSa is originally implemented step by step and all termination criteria are based on certain bounds of those expected footprints of $\bL(U)$. The reader may check that, according to all literature on this topic, there are not any result similar to our Theorem~\ref{condicion suficiente}. In fact, one may see that in \cite[Table II]{saints heegard} (Blahut's paper has not examples of hyperbolic codes) some syndromes \textit{out of} $\A$ are considered.
 
 Usually, abelian codes and AG-codes considered for implement locator decoding are defined over splitting fields, where $q$-orbits are sets with one element. In this case, one may naturally assume that $(0,0)\in \A= \D_{\boldsymbol{\alpha}}(C)$; however, this may fail in case of, for example, binary abelian codes. To overcome this obstacle we introduce our family of ``Hyperbolic-like codes'' which represent a more general situation. There is a pair $\tau\in I$, such that $\tau+\A \varsubsetneq\D_{\boldsymbol{\alpha}}(C)$ (see Example~\ref{ejemplo de (0,t+j)}, where $\tau=(0,13)$).  In detail, we take $\delta\in \N$ and, first, define
 \begin{equation}\label{conjunto B delta}
   \B_\delta= \left\{(i,j)\in I\tq (i+1)(j+1)\leq \delta\right\}\setminus\{(\delta-1,0),(0,\delta-1)\}.
 \end{equation}

Now, we say that the bivariate code $C$ is an \textbf{Hyperbolic like code of designed distance $\delta$}, if there is $\tau\in I$ such that $\tau+\B_\delta=\{\tau+a\mid a\in \B_\delta\}\subset \D_{\boldsymbol{\alpha}}(C)$. We shall see that, in this case, its minimum distance verifies $d(C)\geq \delta$. To do this, we shall use a lower bound to the minimum distance of $C$, called the strong apparent distance, that we denote by $sd^*(C)$. Its definition and description is too large to be reproduced here (see \cite{BBCS2,Camion}); so we only give a brief comment. First, for each codeword $f\in C$ consider the matrix of coefficients of the discrete Fourier transform of $f$; say $M$. The computation of $sd^*(M)$ comes from considering the complementary set of the support of the matrix and it is proved that the weight of any codeword $f$ verifies $w(f)\geq sd^*(M)$. Finally,  $sd^*(C)$  is defined as the minimum of all apparent distances. Then, we get $d(C)\geq sd^*(C)$.

\begin{lemma}
 Let $m=(m_1,m_2)\in I$ and $0\neq M=\left(a_{n}\right)_{n\in I}$ such that $a_n=0$ for all  $n\in m+\B_\delta$, with $\delta\in \N$. Then the strong apparent distance of $M$, denoted by $sd^*(M)$ \cite[Definition 10]{BBCS2}, satisifies $sd^*(M)\geq \delta$. 
\end{lemma}

\begin{proof}
We shall follow the notation in \cite[Remark 11]{BBCS2}. Let $\Row(i)$ the $i$-th row of $M$. Then $\overline{m+\{(0,0),\dots,(0,\delta-2)\}}\subset \Row(m_1)$ modulo $r_1$, for some $m_1\in \Z_{r_1}$. If $\Row(m_1)\neq 0$ then $\epsilon_M(X_1)\geq sd^*(\Row(m_1))\geq \delta$ and hence $sd^*(M)\geq \delta$. So, suppose that $0=\Row(m_1)=\ldots=\Row(m_1+k)$ and $\Row(\overline{m_1+k+1}) \neq 0$, for $0\leq k< \delta-2$. Then $\epsilon_M(X_1)\geq sd^*\left(\Row(\overline{m_1+k+1})\right)\geq\left(\lfloor\frac{\delta}{k+2}\rfloor+1\right)$ and $\omega_M(X_1)\geq k+1$. Then $sd^*(M)\geq \left(\lfloor\frac{\delta}{k+2}\rfloor+1\right)(k+2)= \lfloor\frac{\delta}{k+2}\rfloor(k+2)+(k+2)>\delta$. If $0=\Row(m_1)=\ldots=\Row(\overline{m_1+\delta-2})$ then $\omega_M(X_1)\geq \delta-1$ and we are done.
\end{proof}

\begin{corollary}
  Let $m=(m_1,m_2)\in I$ and $\delta\in \N$. Let $0\neq M=\left(a_{n}\right)_{n\in I}$ such that $a_n=0$ for all  $n\in m+\B_\delta$. Then  $d(C)\geq \delta$.
\end{corollary}

\section{The Berlekamp-Massey-Sakata algorithm}

Let us recall some terminology and some facts about the BMSa. We shall introduce some minor modifications in order to improve its application. 

 We denote by $\N$ the set of natural numbers (including 0) and we define $\Sigma_0=\N\times \N$. We consider the partial ordering in $\Sigma_0$ given by 
 $(n_1,n_2) \preceq (m_1,m_2) \Longleftrightarrow n_1\leq m_1$ and $n_2\leq m_2.$  On the other hand, we will use a (total) monomial ordering \cite[Definition 2.2.1]{Cox}, denoted by ``$\leq_T$'', as in \cite[Section 2]{Sakata 2}. This ordering will be either the lexicographic order (with $X_1>X_2$) \cite[Definition 2.2.3]{Cox} or the (reverse) graded order (with $X_2>X_1$) \cite[Definition 2.2.6]{Cox}. Any result in this paper may be obtained under the alternative lexicographic or graded orders. The meaning of ``$\leq_T$'' will be specified as required.

\begin{definition}\label{Sigmas y Delta rectangulo}
 For $s,k\in\Sigma_0$, we define
\begin{enumerate}
 \item $\Sigma_{s}=\left\{m\in \Sigma_0\tq s\preceq m\right\}$,
 \item  $\Sigma_{s}^k =\left\{m\in \Sigma_0\tq s\preceq m \;\;\text{and}\;\; m <_T k\right\}$ and 
 \item $\Delta_{s}=\left\{n\in \Sigma_0\tq n \preceq s\right\}$.
\end{enumerate}
\end{definition}

 Given $m,n\in \Sigma_0$, we define $m+n$, $m-n$ (provided that $n\preceq m$)  and $n\cdot m$, coordinatewise, as it is usual. An infinite array or matrix is defined as $U=\left(u_n\right)_{n\in \Sigma_0}$; where the $u_n$ will always belong to the extension field $\Le$. In practice, we work with finite arrays defined as infinite doubly periodic ones (see \cite[p. 324]{Sakata 2}) and we consider subarrays, as follows.

\begin{definition}\label{Doubly period array}
Let $U=\left(u_n\right)_{n\in \Sigma_0}$ be an infinite array.
 \begin{enumerate}
  \item We say that $U$ is a doubly periodic array of period $r_1\times r_2$ if the following property is satisfied: for $n=(n_1,n_2)$ and $m=(m_1,m_2)$ we have that $n_i\equiv m_i\mod r_i$ for $i=1,2$ implies that $u_n=u_m$.
  
  \item If $U$ is a doubly periodic array of period $r_1\times r_2$, a finite subarray $u^l\subset U$, with $l\in \Sigma_0$ is the array $u^l=\left(u_m\tq m\in \Sigma_0^l\cap \Delta_{(r_1-1,r_2-1)}\right)$
 \end{enumerate}
\end{definition}

Note that, in the case of period $r_1\times r_2$ we may identify $I=\Z_{r_1}\times\Z_{r_2}=\Delta_{(r_1-1,r_2-1)}$; so that, $u^l=\left(u_m\tq m\in I\right)$ for $l>_T(r_1,r_2)$.

As it is well known, every monomial ordering is a well order, so that any $n\in \Sigma_0$ has a successor. For the graded order we have
$$n+1=\begin{cases}
	    (n_1-1,n_2+1) & \text{if } n_1>0\\
            (n_2+1,0) & \text{if } n_1=0
      \end{cases}.$$

In the case of the lexicographic order, we have to introduce, besides the unique successor with respect to the monomial ordering, another successor that we will only use for the recursion steps over $n\in \Delta_{(r_1-1,r_2-1)}$. We also denote it by $n+1$ as follows:
$$n+1=\begin{cases}
	    (n_1,n_2+1) & \text{if } n_2<r_2-1\\
            (n_1+1,0) & \text{if } n_2=r_2-1
      \end{cases}.$$
      
So, during the implementation of the BMSa (that is, results related with it), the successor of $n\in \Delta_{(r_1-1,r_2-1)}$ will be denoted by $n+1$, independently of the monomial ordering considered.
      
Now we recall some definitions that may be found in \cite[pp. 322-323]{Sakata 2}. For any $f\in \Le[\mathbf{X}]$ or $f\in \Le(r_1,r_2)$, we denote the leading power product exponent of $f$, with respect to ``$\leq_T$'' by $LP(f)$. Of course $LP(f)\in \Sigma_0$. For $F\subset \Le[\mathbf X]$, we denote $LP(F)=\{LP(f)\tq f\in F\}$. 

\begin{definition}
 Let $U$ be a doubly periodic array, $f\in \Le[\mathbf X]$, $n\in \Sigma_0$ and $LP(f)=s$. We write $f=\sum_{m\in supp(f)}f_m\mathbf{X}^m$ and define 
 \[f[U]_n=\begin{cases}
           \displaystyle{ \sum_{m\in\supp(f)}f_m u_{m+n-s}}& \text{if }  n\in\Sigma_s\\
           0 & \text{otherwise}
          \end{cases}.\]

The equality $f[U]_n=0$ will be called a \textbf{linear recurring relation} and in this case, we will say that the polynomial \textbf{$f$ is valid for $U$ at $n$}.
\end{definition}

\begin{definition} Let $U$ be a doubly periodic array and $f\in \Le[\mathbf X]$ with $LP(f)=s$.
 \begin{enumerate}
   \item We say that $f$ generates $U$ and write $f[U]=0$, if $f[U]_n=0$ at any $n\in \Sigma_0$.
   \item For any $u=u^k\subset U$, we say that $f$ generates $u$ if $f[U]_n=0$ at every $n\in \Sigma_s^k$ and we write $f[u]=f[u^k]=0$. In case $\Sigma_s^k=\emptyset$ we define $f[u]=0$.
    \item For any $u=u^k\subset U$, we say that $f$ generates $u$, up to $l<_T k$, if $f[u^l]=0$. 
    
 \item Let $u=u^k\subset U$.
 \begin{enumerate}
  \item  We write the set of generating polynomials for $u$ as
  \[\bL(u)=\{f\in \Le[\bf X]\tq f[u]=0\}.\]
\item  We write the set of generating polynomials for $U$ as
  $$\bL(U)=\left\{f\in\Le[\bf X] \tq f[U]=0\right\},$$
  which was originally called $VALPOL(U)$ \cite[p. 323]{Sakata 2}.
 \end{enumerate}

 \end{enumerate}
\end{definition}

 \begin{remark}\label{hechos sobre lambda de U}
 By results in \cite{Blah}, \cite{Sakata 2} and \cite{Sakata} we have the following facts:
 
\begin{enumerate}
  \item $\bL(U)$ is an ideal of $\Le[\mathbf X]$.
  \item Setting $\overline{\bL(U)}=\left\{\overline{g}\mid g\in \bL(U)\right\}$, and viewing the elements of $\Le(r_1,r_2)$ as polynomials, we have that the ideal $\overline{\bL(U)}=\Le(r_1,r_2)\cap \bL(U)$.
 \end{enumerate}
\end{remark}

Let $0<d\in\N$ and consider the sequence $s^{(1)},\dots,s^{(d)}$ in $\Sigma_0$ satisfying
\begin{equation}\label{desigualdades de los puntos de def}
 s^{(1)}_1>\ldots>s^{(d)}_1=0 \quad\text{and}\quad 0=s^{(1)}_2<\ldots<s^{(d)}_2.
\end{equation}

Now we set 
\begin{eqnarray}\label{los delta para G} \nonumber
\Delta_i&=&\left\{m\in \Sigma_0\tq m\preceq \left(s^{(i)}_1-1,s^{(i+1)}_2-1\right)\right\}_{1\leq i\leq d-1}\\
&=&\Delta_{(s^{(i)}_1-1,s^{(i+1)}_2-1)}
\end{eqnarray}
 and define  $\Delta=\bigcup_{i=1}^{d-1}\Delta_i$, which is called a \textbf{$\Delta$-set} or delta-set, and the elements $s^{(1)},\dots,s^{(d)}$ are called its \textbf{defining points}.

We denote by $\fF$ the collection of sets $F=\left\{f^{(1)},\dots,f^{(d)}\right\}\subset \F[\mathbf X]$ where $\{LP(f^{(i)})=s^{(i)}\tq i=1,\dots, d\}$ satisfy the condition~\eqref{desigualdades de los puntos de def}. We shall say that the elements $F\in\fF$ are of type $\Delta$ and we denote by $\Delta(F)$ the $\Delta$-sets determined by them.

 \begin{definition}\label{base de u}
 Let $U$ be doubly periodic and $u=u^k\subset U$.
 We say that the set $F=\left\{f^{(1)},\dots,f^{(d)}\right\}$ is a minimal set of polynomials for $u$ if:
  \begin{enumerate}
   \item $F\subset \bL(u)$.
   \item $F\in \fF$; that is $\Delta(F)$ exists.
   \item If $g\in\F[\mathbf X]$ verifies $LP(g)\in \Delta(F)$ then $g\not\in\bL(u)$ (i.e. $g[u]\neq 0$).
  \end{enumerate}
 \end{definition}
 
 We denote by $\fF(u)$ the collection of the minimal sets of $u$ and we call $\Delta(F)$ the \textbf{footprint of $\bL(u)$} (even $\bL(u)$ is not an ideal). For any minimal set of polynomials $F=\left\{f^{(1)},\dots,f^{(d)}\right\}$  one may see that, for $i\in \{1,\dots,d-1\}$, the sets $\Delta_i$ in \eqref{los delta para G} are nonempty and they are determined by certain polynomials that we call $g^{(i)}$. In fact, in each iteration, one may construct a set $G=\{g^{(i)}\tq i=1,\dots,d-1\}$ \cite[p. 327]{Sakata 2}.

 \subsection{The algorithm.}

 From \cite{Blah,Sakata 2,Sakata},  we have the following facts: 
 \begin{remark}\label{los delta conjuntos y los minimales} Let $U$ be a doubly periodic array.
 \begin{enumerate}

 \item For any $l\in\Sigma_0$, $u^l\subset U$ and $F,F'\in \fF(u^l)$ we have that $\Delta(F)=\Delta(F')$, so that we may write $\Delta(u^l)$.
  \item $\Delta(u^l)\subseteq \Delta(U)$ for all $l\in \Sigma_0$ and if $k<_T l\in\Sigma_0$  then $\Delta(u^k)\subseteq \Delta(u^l)$. 
  \item For any $l\in\Sigma_0$, the set $\Delta(u^l)$ always exists.
  \item The set $\Delta(U)$ is exactly the footprint (see \cite[p. 1615]{Blah}) of $\bL(U)$, and it is completely determined by any of its Groebner basis.
  \item For any $F\in \fF(u^l)$ we have $F\subset \bL(U)$ implies $\langle F\rangle =\bL(U)$. In fact, $F$ is a Groebner basis for $\bL(U)$ by Definition~\ref{base de u}(3)  and \cite[Definition 2.5]{Cox}.
  \item For any $F\in \fF(u^l)$, we always may construct a ``normalized set'' $F'\in \fF(u^l)$; that is, satisfying the following property: for any $f\in F'$ and for all $m\in \supp(f)\setminus\{LP(f)\}$ we have $m \preceq LP(f')$, for all $f'\in F'$; that is, $m\in \Delta(u^l)$ \cite[Section 6]{Sakata 2}. 
  \item As we have commented, for any $\boldsymbol{\alpha}\in\R$, the equality $\D_{\boldsymbol{\alpha}}\left(\overline{\bL(U)}\right)=\D_{\boldsymbol{\alpha}}\left(\bL(U)\right)$ holds. Then, by \cite[Proposition 5.3.1]{Cox} or \cite[p. 1617, Theorem]{Blah} we have that $\left|\D_{\boldsymbol{\alpha}}\left(\overline{\bL(U)}\right)\right|=|\Delta(U)|$  (see also \cite[p. 1202]{Sakata}).
  \item If $F$ is a reduced Groebner basis for $\bL(U)$ then $LP(F)\subset I$ and, for any $\boldsymbol{\alpha}\in \R$, $D_{\boldsymbol{\alpha}}(F)=D_{\boldsymbol{\alpha}}\left(\overline{\bL(U)}\right)$
 \end{enumerate}
 \end{remark}

 Each iteration in the BMSa gives us a minimal set of polynomials for $u=u^{l+1}$ from such a set $u^l$ and the $\Delta$-set $\Delta(u^l)$. The construction of $\Delta(u^{l+1})$ is based on the following remark.
 
 \begin{remark}\label{resumen condiciones crece delta}
 Suppose that $f\in F\in \fF(u^l)$, and $f[u]_l\neq 0$. Then, by the Agreement Theorem and Sakata-Massey Theorem in \cite{Blah}, and  Lemma 5 and Lemma 6 in \cite{Sakata 2} one of the two following options hold:
 \begin{enumerate}
  \item $l-LP(f)\in\Delta(u^l)$ and then $LP(f)$ will be a defining point of $\Delta(u^{l+1})$.
  
  \item $l-LP(f)\not\in\Delta(u^l)$ and then, $\Delta(u^{l+1})$ will have at least one more point, $l-LP(f)$ itself; in fact,  $\Delta_{l-LP(f)}\subset \Delta(u^{l+1})$
 \end{enumerate}
\end{remark}
  
Before giving a brief description of the Sakata's algorithm we show some previous basic procedures used in it. 

For a minimal set of polynomials $F=\{f^{(1)},\dots,f^{(d)}\}$ of $\bL(u^l)$, with $LP(f^{(i)})=\left(s_1^{(i)},s_2^{(i)}\right)$, for $i=1,\dots,d$, we set $F_\bL=F\cap\bL(u^{l+1})$ and $F_N=F\setminus F_\bL$. We also consider $G=\{g^{(1)},\dots,g^{(d-1)}\}$, mentioned in the paragraph below Definition~\ref{base de u}.\\

\begin{theorem}
 [Berlekamp procedure.  Lemmas 5,6 in \cite{Sakata 2}] Let $f^{(a)}\in F$ and $g^{(b)}\in G$ such that $f^{(a)}\in \bL(u^l)$, $g^{(b)}\in \bL(u^{k})$, for some $k<_T l\in I$, with $f^{(a)}[u]_l=w_a\neq 0$ and $g^{(b)}[u]_{k}=v_b\neq 0$. 

 We define
 \begin{eqnarray*}
  r_1&=&\max\{s_1^{(a)},l_1-s_1^{(b)}+1\},\\
  r_2&=&\max\{s_2^{(a)},l_2-s_2^{(b+1)}+1\}\;\text{and}\\
  \mathbf{e}&=&\left(r_1-l_1+s_1^{(b)}-1,\,r_2-l_2+s_2^{(b+1)}-1 \right).
 \end{eqnarray*}
 Then, setting $r=(r_1,r_2)$, we have that
 $$ h_{f^{(a)},g^{(b)}}=\mathbf {X}^{r-s^{(a)}}f^{(a)}-\frac{w_a}{v_b}{\bf X}^{\mathbf{e}}g^{(b)}\in\bL(u^{l+1}).$$

\end{theorem}

We note that $s_1^{(b)}$ and $s_2^{(b+1)}$ refer to elements of $F$ and not $G$. Now, we establish two procedures to be used in the algorithm.

\begin{procedure}\cite[Theorem 1]{Sakata 2}.\label{proc 1}
 If $f^{(i)}\in F_N$ and $l\in s^{(i)}+\Delta(u^l)$.
\begin{enumerate}
 \item Find $1\leq j\leq d-1$ such that $l_1<s_1^{(i)}+ s_1^{(j)}$ and    $l_2<s_2^{(i)}+ s_2^{(j+1)}$.
 \item In the set F we replace $f^{(i)}$ by $h_{f^{(i)},g^{(j)}}$ obtained by the Berlekamp procedure. The point $s^{(i)}$ will be a defining point of $\Delta(u^{l+1})$ as well.
\end{enumerate}
\end{procedure}

\begin{procedure}\cite[Theorem 2]{Sakata 2}.\label{proc 2}  If $f^{(i)}\in\F_N$ and $l\not\in s^{(i)}+\Delta(u^l)$ then one consider all the following defining points and constructions $h_{f^{(a)},g^{(b)}}$ to replace $f^{(i)}$ (and, possibly, some elements of $G$) with the suitable new polynomials in order to get a new $F\in\fF(u^{l+1})$.

\begin{enumerate}
 \item $S=\left(l_1-s_1^{(i)}+1,l_2-s_2^{(i+1)}+1\right)$; with $f^{(i+1)}\in F_N$ and $1\leq i<d$. 
 Then find $k\in\{1,\dots,d\}$ such that $s^{(k)}\prec S$ and set  $h_{f^{(k)},g^{(i)}}$.

\item $S=\left(l_1-s_1^{(k)}+1,s_2^{(i)}\right)$; for some $k<d$, with $f^{(k)}\in F_N$ and $s^{(i)}\prec S$. Then set $h_{f^{(i)},g^{(k)}}$.

\item $S=\left(l_1+1,s_2^{(i)}\right)$ with $i < d$. 
Then set $h=X_1^{l_1-s_1^{(i)}+1}\cdot f^{(i)}$.

\item $S=\left(s_1^{(i)},l_2-s_2^{(j)}+1\right)$ for $j>2$ with $f^{(j)}\in F_N$ and $s^{(i)}\prec S$. Then set $h_{f^{(i)},g^{(j-1)}}$.

\item $S=\left(s_1^{(i)},l_2+1\right)$. Then set $h=X_2^{l_2-s_2^{(i)}+1}\cdot f^{(i)}$.
\end{enumerate} 
\end{procedure}

Now, we can show a brief scheme of the Sakata's algorithm. See \cite[p. 331]{Sakata 2} for a detailed description.

\begin{algorithm}[Sakata]\label{algoritmo sakata}
 We start from a finite doubly periodic array, $u\subset U$.
 \begin{enumerate}
  \item[$\qed$] Initialize $|l|=0$; that is $l=(0,0)$, $F=\{1\}$, $G=\emptyset$ and $\Delta=\emptyset$.
 \item[$\qed$] For $l\geq (0,0)$,
 \item For each $f^{(i)}\in F$ for which $f^{(i)}\in F_N$ we do
 \begin{description}
  \item[-] If $l\in s^{(i)}+\Delta(u^l)$ then replace $f^{(i)}$ by Procedure~\ref{proc 1}. 
  \item[-] Otherwise, replace $f^{(i)}$ by one or more polynomials by Procedure~\ref{proc 2}.
 \end{description}
 \item Then form the new $F$, $G$ and $\Delta(u^{l+1})$.
 \item Set $l:=l+1$.
 \end{enumerate}
\end{algorithm}

Let $l\in \Sigma_0$, $F\in \fF(u^l)$ and consider the ideal $\langle F\rangle$ in $\Le[\mathbf{X}]$. We suppose WLOG that the elements in $F$ are written in their normal form. Then, on the one hand, it may happen that $F$ is not a Groebner basis for $\langle F\rangle$; on the other hand, even if $F$ is a Groebner basis for $\langle F\rangle$, it may happen that $F$ is not a Groebner basis for $\bL(U)$. As we have commented in Introduction, in \cite{Cox et al Using,Hackl,rubio,Sakata 2} the reader may find termination criteria based on the shapes of all possible extensions from $\Delta(u^l)$ to $\Delta(U)$. We shall show a new termination criterion based on the existence of the set $\B$. Before this, in the next section, we shall give an improvement of the framework for applying the locator decoding algorithm in such a way that we may consider a translated table.

\section{A new framework for locator decoding.}

Locator decoding in (bivariate) abelian codes was introduced in \cite{Sakata} (see also \cite{Blah}). Let us recall, and extend slightly, the basic ideas.
 
 Let $C$ be a bivariate code over $\F(r_1,r_2)$ with defining set $\D_{\boldsymbol{\alpha}}\left(C\right)$, with respect to some fixed $\boldsymbol{\alpha}\in \R$. Suppose a word $c\in C$ was sended and the polynomial $c+e$ in $\F(r_1,r_2)$ has been received. So that, the polynomial $e$ represents the error that we want to find out. To do this, we define the locator ideal in $\Le(r_1,r_2)$, which is defined originally in $\Le[\mathbf{X}]$ (see \cite{Blah,Sakata}).

 \begin{definition}\label{ideal locator}
  In the setting above, the locator ideal for $e$ is
  \[L(e)=\left\{f\in \Le(r_1,r_2)\tq f(\boldsymbol{\alpha}^n)=0,\;\forall n\in \supp(e)\right\}.\]
 \end{definition}
 
 Having in mind that $\Le|\F$ is a splitting field for $U$, it is easy to see that $\D_{\boldsymbol{\alpha}}\left(L(e)\right)=\supp(e)$. Our objective is to find  the defining set of $L(e)$ and hence $\supp(e)$. The final step (that we will not comment) will be to solve a system of equations to get the coefficients of $e$ (in case $q>2$). To do this, we shall connect $L(e)$ to the linear recurring relations as follows. Based on the so called syndromes of the received polynomial, we are going to determine a suitable doubly periodic array $U=\left(u_n\right)_{n\in \Sigma_0}$ such that the equality $L(e)=\overline{\bL(U)}$ holds (see Remark~\ref{hechos sobre lambda de U}). We begin dealing with syndromes. As it is usual in locator decoding, we first consider (theoretically) the syndrome values of $e\in \F(r_1,r_2)$: let $\tau\in \Z_{r_1}\times \Z_{r_2}$ and define $U=\left(u_n\right)_{n\in \Sigma_0}$, such that $u_n= e\left(\boldsymbol{\alpha}^{\tau+n}\right)$. Clearly, $U$ is an infinite doubly periodic array.

\begin{definition}\label{def de U}
 Let $e\in \F(r_1,r_2)$, $\tau\in \Z_{r_1}\times \Z_{r_2}=I$ and define $U=\left(u_n\right)_{n\in \Sigma_0}$, such that $u_n= e\left(\boldsymbol{\alpha}^{\tau+n}\right)$. We call $U$ the syndrome table afforded by $e$ and $\tau$.
\end{definition}

 In practice, we do not know all values of $U$. Let us return, for a moment, to the error correcting context. By the notion of defining set, for each $\tau+n\in \D_{\boldsymbol{\alpha}}(C)$, one has that $(c+e)\left(\boldsymbol{\alpha}^{\tau+n}\right)=e\left(\boldsymbol{\alpha}^{\tau+n}\right)$; so, the syndrome values of the error polynomial $e$ are known for all elements in $\D_{\boldsymbol{\alpha}}(C)$.

 Now we state the mentioned equality of ideals. The proof of the following theorem is (\textit{mutatis mutandi}) similar to that of \cite[p. 1202]{Sakata}.

\begin{theorem}\label{igualdad ideales}
 Let $U$ be the syndrome table afforded by $e$ and $\tau$. For any $f\in \Le(r_1,r_2)$ the following conditions are equivalent:
 \begin{enumerate}
  \item $f\in L(e)$.
  \item $\sum_{s\in\supp(e)}e_{s}\boldsymbol{\alpha}^{s\cdot n} f\left(\boldsymbol{\alpha}^{s}\right)=0,\;\text{for all}\; n\in \Sigma_\tau$.
  \item $f\in \overline{\bL(U)}$.
 \end{enumerate}
Consequently, $L(e)=\overline{\bL(U)}$.
\end{theorem}

Theorem~\ref{igualdad ideales}, together with Remark~\ref{los delta conjuntos y los minimales}(8), say that if $F$ is a Groebner basis of $\bL(U)$, then $$\D_{\boldsymbol{\alpha}}(L(e))=\D_{\boldsymbol{\alpha}}(\overline{\bL(U)})=\D_{\boldsymbol{\alpha}}(F)$$
 according to the notation of Section II.

The ideal $\bL(U)$ drives us to the framework used in the BMSa in the specific case of $U$, the syndrome table afforded by $e$. 

 \subsection{Obtaining a true Groebner basis for the ideal $\bL(U)$. Updates and sufficient conditions.}

Suppose that, following the BMSa we have constructed for $l\in I$, the foorprint $\Delta(u^l)$ and the minimal set of polynomials $F_l\in \fF(u^l)$. In this section, we prove that, under the assumption $\omega(e)\leq 4$, if $\l\not\in \B$ then $F_l=F_{l+1}$ and, if $l\in \B$ is its maximum, then $\langle F_{l+1}\rangle=\bL(U)$; that is, the normal form of the elements of $F_{l+1}$ is a Groebner basis for $\bL(U)$.

\begin{lemma}\label{cota para l_1+1 por l_2+1}
 Let $U$ be the syndrome table afforded by $e$ and $\tau$, with $\omega(e)\leq t\leq 4$. Suppose that, following the BMSa we have constructed, for $l=(l_1,l_2)$, with $u=u^l$, the sets  $\Delta(u)=\Delta$ and $F\in\fF(u)$.  We also suppose that there is $f\in F$ such that $f[u]_l\neq 0$ and that $l\not\in LP(F)+\Delta$; that is, the delta-set will increase (see Remark~\ref{resumen condiciones crece delta}(2)).
 Then $$(l_1+1)(l_2+1)\leq 2t+1.$$
\end{lemma}
\begin{proof}
We shall prove the result for $t=4$. The other cases are similar and simpler than this. Suppose that $F=\left\{f^{(1)},\dots,f^{(d)}\right\}$ with $LP(f^i)= s^{(i)}$ for $i=1,\dots, d\geq 2$. Setting $f=f^{(i)}$ we have, by hypothesis, $f^{(i)}[u]_l\neq 0$ and $l\not\in s^{(i)}+\Delta$.
	
First note that $|\Delta|\leq 3$ because the size will be increased. So let us list all possible delta-sets: $\Delta_{11}=\left\{(0,0)\right\}$, $\Delta_{21}=\left\{(0,0),(0,1)\right\}$, $\Delta_{22}=\left\{(0,0),(1,0)\right\}$, $\Delta_{31}=\left\{(0,0), (0,1),(0,2)\right\}$, $\Delta_{32}=\left\{(0,0),(1,0),(0,1)\right\}$ and $\Delta_{33}=\left\{(0,0),(1,0),(2,0)\right\}$.

We also note that, by definition of delta-set, if $l\not\in LP(F)+\Delta$ then $\Sigma_l\cap (LP(F)+\Delta)=\emptyset$.
	
\textbf{Case a}: $l_1>6$. By paragraph above, we only have to consider $l_1=7$. As $s^{(1)}_1\leq 3$, we have that $l_1-s^{(i)}_1\geq 7-3=4$, thus, at least $(3,0),(4,0)$ increase $\Delta(u^{l+1})$, which is impossible. So we should have $l_1\leq 6$.

\textbf{Case b}: $l_1=6$ and $l_2\geq 1$. Again, we only have to consider $l=(6,2)$. Then, the points $(3,0)$ and $(3,1)$ will be added. If $|\Delta|=2$ then we have to add, in addition, $(2,0)$ and $(2,1)$, and for $\Delta=\Delta_{11}$ we have to add besides the points below, $(1,0)$ and $(1,1)$. In all cases we get $|\Delta(u^{l+1})|>4$, which is impossible.

\textbf{Case c}: $l_1=5$ and $l_2\geq 1$, so we set $l=(5,1)$. If $s^{(1)}_1=3$ and $i=1$ then $(2,1)\in \Delta(u^{l+1})$ which implies that $(0,1),(1,1)\in \Delta(u^{l+1})$ too. In case $i=2$, then at least $(3,0)$ and $(4,0)$ will be added. If $s^{(1)}_1=2$ then $l_1-s^{(i)}_1\geq 3$ so that, for $i=1$ we have that $(2,0),\;(2,1),\;(3,0),\;(3,1)\in  \Delta(u^{l+1})$; for $i=2$ then $l_1-s^{(i)}_1\geq 4$, so $(2,0),\;(3,0),\;(4,0)\in  \Delta(u^{l+1})$.  The case $s^{(1)}_1=1$ is trivial and then in all cases we get $|\Delta(u^{l+1})|>4$, which is impossible.
	
\textbf{Case d}: $l_1=4$ and $l_2\geq 1$, so that, set $l=(4,1)$. If $s^{(i)}_1=3$ then we have to add at least $(0,1)$ and $(1,1)$, if $s^{(i)}_2=2$ then we must have $i=2$ and we should add at least $(1,1)$ and $(2,1)$, for $\Delta_{32}$ and $(0,1)$ in addition, for $\Delta_{21}$. For $s^{(i)}_1=1$ then $(1,l_2-s^{(d)}_2)$, $(2,l_2-s^{(d)}_2)$ and $(3,l_2-s^{(d)}_2)$ should be added. All of them are impossible.

\textbf{Case e}:  $l_1=3$ and $l_2\geq 2$; so that $l=(3,2)$. If $i=d$ then we add at least $(2,l_2-s^{(d)}_2)$ and $(3,l_2-s^{(d)}_2)$ for those $|\Delta|=3$ and, in addition, $(0,l_2-s^{(d)}_2)$ and $(1,l_2-s^{(d)}_2)$ for those $|\Delta|\leq 2$. For $\Delta_{32}$ and $i=2$, we have $l-s^{(2)}=(l_1-s^{(2)}_1,1)$ so we add at least $(1,1)$ and $(2,1)$. Finally, the case $i=1$ is obvious and then in all cases we get $|\Delta(u^{l+1})|>4$, which is impossible. 

\textbf{Case f}: $l_1=2$ and $l_2\geq 3$. Take $l=(2,3)$ and repeat  \textbf{Case~e} changing $l_2$ by $l_1$; $s^{(d)}_2$ by $s^{(1)}_1$ and so.

\textbf{Case g}: $l_1=1$ and $l_2\geq 4$. Take $l=(1,4)$ and repeat \textbf{Case~d} with the adecuate changes, as above.

The last case, $l_1=0$ is immediate by Procedure~\ref{proc 2}.
\end{proof}

Lemma above says us that, for $\delta=2t+1$ if $f[u]_l\neq 0$ and that $l\not\in LP(F)+\Delta$ then $l\in \A$. Now we get $\B_{2t+1}$ by studying, among others, the pairs $(2t,0)$ and $(0,2t)$.

\begin{lemma}
 Let $U$ be the syndrome table afforded by $e$ and $\tau$, with $\omega(e)\leq t\leq 4$.  Suppose that, following the BMSa we have constructed, for $l=(l_1,l_2)$, with $u=u^l$ the sets  $\Delta(u)=\Delta$ and $F\in\fF(u)$.  If $l_k>2t-1$, for $k\in \{1,2\}$ then $f[u]_l=0$.
\end{lemma}
\begin{proof}
Immediate from the fact that $l_1-s^{(1)}_1>t$ or $l_2-s^{(d)}_2>t$
\end{proof}

 The proof of the next lemma is a direct computation similar to that used in the previous one.

\begin{lemma}\label{puntos fuera de la escalera}
 Let $U$ be the syndrome table afforded by $e$ and $\tau$, with $\omega(e)\leq t\leq 4$.  Suppose that, following the BMSa we have constructed, for $l=(l_1,l_2)$, with $u=u^l$ the sets  $\Delta(u)=\Delta$ and $F\in\fF(u)$.  If $l=(l_1,l_2)$ is such that $(l_1+1)(l_2+1)> 2t+1$ then $l\not\in LP(F)+\Delta$ and hence $\Sigma_{l}\cap (\Delta+LP(F))=\emptyset$.
 
 Thus, if $n\in \Sigma_{(i,j)}$, with $(i,j)\in I$ satisfying $(i+1)(j+1) > 2t+1$ then $f[u]_n = 0$, for any $f\in F$.
\end{lemma}

Let us summarize the results above in the following theorem.

\begin{theorem}\label{puntos fuera de la escalera global}
 Let $U$ be the syndrome table afforded by $e$ and $\tau$, with $\omega(e)\leq t\leq 4$. Suppose that, following the BMSa we have constructed, for $l=(l_1,l_2)$, and $u=u^l$, the sets  $\Delta(u)=\Delta$ and $F\in\fF(u)$. For any $f\in F$, we have that:
\begin{enumerate}
 \item If $f\in F$ is such that $f[u]_l\neq 0$ and $l\not\in LP(f)+\Delta$ then $(l_1+1)(l_2+1)\leq 2t+1$.
 \item  If $l_k>2t-1$, for $k\in \{1,2\}$ then $f[u]_l=0$.
 \item If $l_1,l_2\neq 0$ and $(l_1+1)(l_2+1)>2t+1$ then $f[u]_l= 0$ for any $f\in F$.
\end{enumerate}
Hence, if $l\not\in\B_{2t+1}$ then $f[u]_l=0$.
\end{theorem}

 Let $C$ be a bivariate code with (a lower bound of its) error-correction capability $t\leq\lfloor\frac{d(C)-1}{2}\rfloor$ and let $g=c+e$ the received polynomial. Let $U$ be the syndrome table afforded by $e$ and $\tau\in I$, and assume that $\omega(e)\leq t\leq 4$. Suppose that $\overline{\tau+\B_{2t+1}} \subset \D_{\boldsymbol{\alpha}}(C)$ (see Equality~\eqref{conjunto B delta}). Then, for all $l\in \B_{2t+1}$, the values $u_l= e\left(\boldsymbol{\alpha}^{\tau+l}\right)=g\left(\boldsymbol{\alpha}^{\tau+l}\right)$ are known.

\begin{theorem}\label{condicion suficiente}
Let $C$ be a bivariate code with (a lower bound to its) error-correction capability $t\leq 4\leq \lfloor\frac{d(C)-1}{2}\rfloor$. Suppose $\overline{\tau+\B_{2t+1}} \subset \D_{\boldsymbol{\alpha}}(C)$, for some $\tau\in J$ and $u_{(0,j)}\neq 0$, for some $j<t$ (respectively if $u_{(i,j)}\neq 0$ with $i+j=1$).   Then any transmision of codewords of $C$ with no more than $t$-errors may be decoded by applying the BMSa with the lexicographic order (respectively the graded order) on $\B_{2t+1}$.
\end{theorem}
\begin{proof}
 We begin by considering the lexicographic order.
 
 We recall that at initializing the BMSa we take $F=\{1\}$, so that $1[u]_{(0,j)}=u_{(0,j)}$ and the first two defining points are $(1,0)$ and $(0,j+1)$, which indicate us the necessity $u_{(0,j)}\neq 0$ for some $j<t$.
 
 Now, to do all steps for the pair of the form $(0,*)$ we have to compute at most $l=(0,j)$ for $j=0,\dots2t-1$. Now suppose we have compute $\Delta(u^l)$ for all $l=(l_1,l_2)$ with $l_1,l_2\neq 0$ and $(l_1+1)(l_2+1)\leq 2t+1$, which is equivalent for $t\leq 4$  to the values $l_2=0,\dots,t-l_1$. Then any step considered after that, say again $l$, must verify $f[u]_l=0$, by Theorem~\ref{puntos fuera de la escalera global}(3).
 
 Clearly, the last point for which our $\Delta$ may be increased is $(2t-1,0)$. After that, Theorem~\ref{puntos fuera de la escalera global} guarantees us that $\Delta$ cannot increase their size. However, any step of the form $l=(l_1,0)$ with $l_1\leq 2t-1$ may satisfy $l\in LP(F)+\Delta$ and so $F$ may be changed. So we have to consider them.
 
 For any step of the form $l\geq_T(2t,0)$ it happens that $l\not\in LP(F)+\Delta$ and clearly $f[u]_l=0$, for all $f\in F\in \fF(u^l)$ because $|\Delta(u^{l+1})|\leq t$.
 
 Now we deal with graded order. Suppose we compute $\Delta(u^l)$ for all $\{l=(l_1,l_2)\tq l_1+l_2\leq t\}$, and $F$ is the minimal set of polynomials obtained in the last iteration, with $\Delta(F)=\Delta$.  Consider a point $l=(l_1,l_2)$ such that $l_1+l_2\geq t+1$. Then one may check that $(l_1+1)(l_2+1)>2t+1$, for $t\leq 4$; so, if  one has that $l_1,l_2\neq 0$ then Theorem~\ref{puntos fuera de la escalera global} says that  $f[u]_l=0$ for all $f\in F$.
Finally, it may happen that $f[u]_l\neq 0$ for $l=(a,0),(0,a)$ with $a \in \{t+1,\dots,2t-1\}$ (the cases $l=(j,0)$, with $j\geq 2t$ has been already seen). We will continue forming minimal sets of polynomials until consider all of them.
 
 Therefore in any of the monomial orders considered, the polynomials of $F$ are valid in $I$; so that $F\subset \bL(U)$ and then $\langle F\rangle =\bL(U)$. By Remark~\ref{los delta conjuntos y los minimales}(5) and Theorem~\ref{igualdad ideales} we are done.
\end{proof}

\begin{example}\label{ejemplo de (0,t+j)}
 Consider the code $C$, in $\F_2(5,15)$ with primitive root $a$, and $\D_{(\alpha,\beta)}(C)=Q(0,13)\cup Q(1,13)\cup Q(2,13)\cup Q(3,13)\cup Q(4,13)\cup Q(0,0)\cup Q(0,1)$. One may check that the strong apparent distance  $sd^*(C)=6$, so that $t=2$ is a lower bound for the error correction capability of $C$. For the error polynomial $e=X_1^2X_2^2+X_2$ and $\tau=(0,13)$ we have the first value $u_{(0,0)}=e(\alpha^0,\beta^{13})=a^{4}$ and the last one $u_{(4,0)}=e(\alpha^4,\beta^{13})=a^2$. So that we arrange
  \[\left(u_n\tq n\in\B_5\right)=\begin{pmatrix}
        a^{4} & a^{2} & 0 & a^{5}\\
        a^{14} & a^9 \\
        a^{3} \\
        a^2
      \end{pmatrix}.\]

Next table summarizes all computation with respect to the lexicographic order.
\begin{footnotesize}
 \[\begin{array}{|l|l|l|l|}\hline
 l&F\subset \bL(u^{l+1})&G&\Delta(u^{l+1})\\ \hline
  \text{Initializing}&\{1\}&\emptyset&\emptyset\\ \hline
(0,0)\rightarrow&\{X_1,X_2\}&\{1\}&\{(0,0)\}\\ \hline
(0,1)\rightarrow&\{X_1,X_2+a^{13}\}&\{1\}&\{(0,0)\}\\ \hline
 (0,2)\rightarrow&\{X_1,X_2^2+a^{13}X_2+a^{11}\}&\{X_2+a^{13}\}&\{(0,0),(0,1)\}\\ \hline
  (0,3)\rightarrow&\{X_1,X_2^2+a^{5}X_2+a^3\}&\{X_2+a^{13}\}&\{(0,0),(0,1)\}\\ \hline
(1,0)\rightarrow&\begin{array}{l}
                  \{X_1+a^6X_2+a^2,\\X_2^2+a^{5}X_2+a^3\}
                 \end{array}
&\{X_2+a^{13}\}&\{(0,0),(0,1)\}\\ \hline
(1,1)\rightarrow &
\begin{array}{l}
\{X_1+a^8X_2+a^7,\\
X_2^2+a^{5}X_2+a^3\}
\end{array}
&\{X_2+a^{13}\}&\{(0,0),(0,1)\}\\ \hline
  (2,0),(3,0)\rightarrow&
\text{Same} & \text{Same} & \text{Same}
\\\hline
\end{array}\]
\end{footnotesize}
The reader may check that $\D_{\boldsymbol{\alpha}}(\bL(U))=\D_{\boldsymbol{\alpha}}(\langle F\rangle)=\{(2,2),\;(0,1)\}$.
\end{example}

%
%
 \bibliographystyle{splncs04}
%

\end{document}